\documentclass[11pt]{article}
\usepackage[margin=1in]{geometry}

\usepackage[utf8]{inputenc}
\usepackage{amsmath,amssymb,amsfonts,amsthm,bm}
\usepackage{latexsym,bbm,graphicx,float}
\usepackage{mathtools,braket}
\usepackage[multiple]{footmisc} 
\usepackage{enumitem}
\usepackage{multirow, multicol}
\usepackage{thmtools, thm-restate}
\usepackage[plain]{algorithm}
\usepackage{algpseudocode}
\usepackage[usenames,dvipsnames,svgnames,table]{xcolor}
\usepackage{subcaption}
\usepackage{soul}
\usepackage{varwidth}
\usepackage{url}
\usepackage{soul} 
\usepackage{todonotes}
\usepackage{titlesec} 

\usepackage[backend=biber,backref=true,backrefstyle=two,style=alphabetic,maxalphanames=4,maxcitenames=2,maxbibnames=99,natbib=true]{biblatex}
\usepackage[breaklinks,hyperfootnotes=false]{hyperref}
\hypersetup{
	colorlinks   = true, 
	urlcolor     = Maroon, 
	linkcolor    = Maroon, 
	citecolor   =  Maroon 
}
\usepackage{cleveref}

\newtheorem{theorem}{Theorem}

\theoremstyle{definition}

\newtheorem{claim}{Claim}

\newtheorem{definition}{Definition}

\DeclareCaptionSubType*{table}

\algnewcommand\algprocedure{\textbf{Procedure:}}
\algnewcommand\Procedurename{\item[\underline{\algprocedure}]}
\algnewcommand\algmain{\textbf{Main:}}
\algnewcommand\Main{\item[\underline{\algmain}]}
\algnewcommand\algorithmicinput{\textbf{Input:}}
\algnewcommand\Input{\item[\algorithmicinput]}
\algnewcommand\algorithmicoutput{\textbf{Output:}}
\algnewcommand\Output{\item[\algorithmicoutput]}

\makeatletter
\renewenvironment{proof}[1][\proofname] {\pushQED{\qed}\normalfont\topsep\z@\@plus0\p@\relax\trivlist\item[\hskip\labelsep\bfseries#1\@addpunct{:}]\ignorespaces}{\popQED\endtrivlist\@endpefalse}
\makeatother

\makeatletter
\def\thm@space@setup{%
	\thm@preskip=8pt plus 2pt minus 4pt
	\thm@postskip=\thm@preskip 
}
\makeatother
\makeatletter
\renewcommand{\ALG@beginalgorithmic}{\small}
\makeatother

\titlespacing\section{0pt}{12pt plus 2pt minus 2pt}{-1pt plus1pt minus 1pt}
\titlespacing\subsection{0pt}{12pt plus 2pt minus 2pt}{-1pt plus 1pt minus 1pt}
\titlespacing\subsubsection{0pt}{12pt plus 2pt minus 2pt}{-1pt plus 1pt minus 1pt}

\setlist{topsep = 0pt plus1pt}

\newcommand{\proofcase}[1]{\par\smallskip\noindent{\textbf{#1}}}
\renewcommand{\hat}{\widehat}

\newcommand{\bg}{{\mathbb G}}

\newcommand{\mc}{{\mathcal C}}

\newcommand{\ml}{{\mathcal L}}
\newcommand{\mm}{{\mathcal M}}

\newcommand{\calC}{\mathcal{C}}
\newcommand{\unrev}{\textrm{unrev}}
\newcommand{\rev}{\textrm{rev}}
\newcommand{\sig}{\textrm{sig}}
\newcommand{\extsig}{\textrm{extsig}}
\newcommand{\optsig}{\textrm{sig}^\textrm{opt}}
\newcommand{\equivsig}{=}
\newcommand{\calM}{\mathcal{M}}

\newcommand{\red}[1]{{\color{red} #1 }}

\algnewcommand{\LeftComment}[1]{\State \(\triangleright\) #1}

\newenvironment{proof*}[1][\proofname]{
	
	\begin{proof}[#1]}{\end{proof}}

\allowdisplaybreaks

\def\npoSecProofsinApp{0}

\title{Necessarily Optimal One-Sided Matchings}

 \author{
    Hadi Hosseini\footnote{College of Information Sciences and Technology, Penn State University.\ \ttfamily hadi@psu.edu
 	}
 	\and
 	Vijay Menon\footnote{David R.\ Cheriton School of Computer Science, University of Waterloo.\ \ttfamily vijay.menon@uwaterloo.ca}
 	\and 
 	Nisarg Shah\footnote{Department of Computer Science, University of Toronto.\ \ttfamily nisarg@cs.utoronto.ca
 	}
 	\and
 	Sujoy Sikdar\footnote{Department of Computer Science, Binghamton University.\ \ttfamily ssikdar@binghamton.edu
 	}
}

\date{}

\begin{document}
	
	\maketitle
	\begin{abstract}
 We study the classical problem of matching $n$ agents to $n$ objects, where the agents have ranked preferences over the objects. We focus on two popular desiderata from the matching literature: Pareto optimality and rank-maximality. Instead of asking the agents to report their complete preferences, our goal is to learn a desirable matching from partial preferences, specifically a matching that is necessarily Pareto optimal (NPO) or necessarily rank-maximal (NRM) under any completion of the partial preferences. We focus on the top-$k$ model in which agents reveal a prefix of their preference rankings. We design efficient algorithms to check if a given matching is NPO or NRM, and to check whether such a matching exists given top-$k$ partial preferences. We also study online algorithms for eliciting partial preferences adaptively, and prove bounds on their competitive ratio. 
\end{abstract}

\section{Introduction}\label{sec:intro}
Resource allocation is a fundamental problem in artificial intelligence and multi-agent systems. One particular type of resource allocation deals with assigning a number of indivisible objects to agents according to their preferences. These problems have given rise to a wide array of research in economics~\citep{moulin2004fair} and theoretical computer science~\citep{david2013algorithmics}. 


The focus of our work is the special case of allocating $n$ objects to $n$ agents (so each agent is matched to a single object), which models many real-world applications (see, e.g., \cite{hylland1979efficient,abdulkadirouglu1999house}). For instance, imagine allocating office spaces to faculty members in a building. Instead of asking each faculty member to report a full preference ranking over the available offices, the department head may ask them to reveal their top choices, and then, if need be, she may ask individual faculty members to reveal their next best choices, and so on. The goal is to come up with a matching that \emph{necessarily} satisfies some form of ``economic efficiency'' while asking as few queries as possible. That is, it must satisfy the economic efficiency property regardless of the parts of the preferences that were not elicited, because otherwise some faculty may find the final matching undesirable.

What form of economic efficiency might one want for a matching? Luckily, decades of research in matching theory offers two prominent desiderata: \emph{Pareto optimality}~\cite{shapley1974cores,cirillo2012economics} and \textit{rank-maximality}~\cite{irving2003greedy,irving2006rank}. Informally, Pareto optimality requires that no other matching be able to make some agents better off without making any agent worse off.
Despite its wide use, Pareto optimality may be a weak guarantee in many practical settings. Hence, a stronger guarantee called rank-maximality is often used in practical applications such as assigning papers to referees~\cite{garg2010assigning}, assigning rented resources to customers~\cite{abraham2006assignment}, and assigning students to schools~\cite{abraham2009matching}. Informally, rank-maximality requires matching as many agents to their top choice as possible, subject to that matching as many agents to their second choice as possible, and so on. 

The study of these axioms, along with axioms of fairness and incentive-compatibility, has led to the design of elegant rules such as random priority and probabilistic serial~\cite{bogomolnaia2001new,che2010asymptotic}, which have been made easily accessible by not-for-profit endeavors such as MatchU.ai (\url{www.matchu.ai}). However, in many real-life situations, one rarely has access to complete preferences which can simply be fed to these rules. More often than not, the problem at hand is to come up with a desirable matching of agents to objects with only partial information about agents' preferences. While the role of partial preferences has been well-explored in the sister problem known as \emph{two-sided matching}~\cite{rastegari2013two,liu2014stable}, in which agents are matched to other agents (for example matching roommates, kidney donors to patients, or men to women), it has been significantly understudied for one-sided matching. This is the primary focus of our work. 
In particular, we consider the following  three research questions.
\begingroup
\begin{quote}
    \emph{Given partial preferences, can we efficiently check if a given matching is NPO or NRM? Can we efficiently check if an NPO or NRM matching exists (and if so, find one)? 
    How much online elicitation of preferences is needed to find an NPO or NRM matching?
    }
\end{quote}
\endgroup

\subsection{Our Results}\label{sec:results}
We study the \emph{next-best} query model which allows, in each query, asking one agent to reveal her next best choice. Partial preferences elicited under this model are referred to as \emph{top}-$k$ preferences as each agent reveals a prefix of her preference ranking consisting of the top $k$ items (where $k$ can be different for each agent). As mentioned above, our goal is to design online algorithms that ask few queries and learn a matching that is necessarily Pareto optimal or necessarily rank-maximal w.r.t.\ the elicited partial preferences. We study their performance in terms of their competitive ratio~\cite{borodin2005online}. 

Note that these algorithms need to check if there already exists a matching that is NPO or NRM w.r.t.\ the elicited preferences in order to decide if further elicitation is required. Along the way, it is also useful to know if a \emph{given} matching is NPO or NRM w.r.t.\ the given partial preferences. Thus, we also study these two questions for NPO and NRM. 

For necessary Pareto optimality, we show that checking whether a given matching is NPO reduces to computing the maximum weight matching in a bipartite graph, which can be solved efficiently. We also show that there exists an NPO matching if and only if there exists a matching that matches at least $n-1$ agents to objects they have revealed, where $n$ is the number of agents; this can be checked in polynomial-time, and if this is the case, our algorithm efficiently finds an NPO matching. Our main result is an online algorithm for finding an NPO matching that has a competitive ratio of $O(\sqrt{n})$. By proving a matching lower bound, we show that our algorithm is asymptotically optimal. 

For necessary rank-maximality, we design efficient algorithms for checking whether a given matching is NRM. For checking whether an NRM matching exists; unlike in the case of NPO, these questions do not seem to reduce to checking a simple analytical condition. In the supplementary material, we show that a simple characterization can be obtained in the special case where each agent has revealed exactly the same number of objects. For online elicitation, we show a constant lower bound on the competitive ratio, and conjecture that the optimal competitive ratio is in fact constant, unlike in the case of NPO.

\subsection{Related Work}\label{sec:related}
The problem of matching $n$ agents to $n$ objects models a number of scenarios and is well-studied, sometimes as the house allocation or one-to-one object allocation problem (e.g., \cite{hylland1979efficient,abdulkadirouglu1998random,abra04,irving2003greedy}). There are two approaches to learning a desirable matching given only partial information about agents' preferences. One approach is to consider the worst case scenario over missing preferences; imposing a property $X$ in the worst case leads to the ``necessarily $X$'' concept. This concept has been considered in other contexts such as in voting \cite{xia08} and in one-sided matching problems where an agent may be assigned to more than one good \cite{bou10,aziz15,aziz19}. However, in the latter case, even having access to agents' full preference rankings over individual goods is not sufficient for checking Pareto optimality, which leads \citet{aziz19} to consider NPO given full rankings. In contrast, in our one-to-one matching setting, full preference rankings over individual goods are sufficient to check whether a given matching is PO, which leads us to consider partial (top-$k$) rankings.

Another approach is to assume that the underlying preferences are drawn from a known distribution~\cite{rastegari2013two}, and consider the probability of satisfying $X$~\cite{aziz2019pareto}. We note that asking whether a matching satisfies $X$ with probability $1$ often coincides with the ``necessarily $X$'' concept. 

We focus on the top-$k$ model of partial preferences, which is commonly used in the literature~\cite{aziz2015possible,drummond2013elicitation}. However, there are also other interesting models; one such model is where the information available regarding each agent's preferences is in the form of a set of pairwise comparisons among objects~\cite{aziz2019pareto}. In general, the goal of this line of work is to make matching algorithms more practically relevant, to which our work contributes. 


Other than economic efficiency, researchers have also focused on notions of fairness and incentive-compatibility in matching agents to objects, which has led to the design of rules such as random priority and probabilistic serial~\cite{bogomolnaia2001new,abdulkadirouglu1998random,hosseini2018investigating}. A related line of research considers two-sided matching markets, where agents are matched to other agents. Working with partial preferences is well-explored in such markets~\cite{rastegari2013two,liu2014stable}. 


	\section{Preliminaries} \label{sec:prelims}
Let $N = \{a_1, \cdots, a_n\}$ denote the set of agents and $O = \{o_1, \cdots, o_n\}$ be the set of objects.
We use $[k]$ to denote the set $\{1,\dots, k\}$. 
Thus, for $i\in [n]$, $a_i$ and $o_i$ represent agent and object $i$ respectively.
A \textit{matching} of agents to objects is a bijection $M: N\to O$. We sometimes refer to a matching as a collection of pairs $\set{(a_i, M(a_i) \mid i \in [n])}$ and use $\mathcal{M}$ to denote the set of all possible matchings.


Each agent $a_i \in N$ has an underlying linear order $R_i$ over $O$, where a linear order over $O$ is a transitive, antisymmetric, and total relation on $O$. We use $\ml(O)$ to denote the set of all linear orders over $O$, and refer to $R = (R_1, \cdots, R_n) \in \ml(O)^n$  as a \textit{(complete) preference profile}. Additionally, for $i, j, k \in [n]$, if $a_i$ prefers $o_j$ to $o_k$ according to $R_i$, then we denote this by $o_j \succ_{R_i} o_k$ or $o_k \prec_{R_i} o_j$. 

We formally describe two notions of economic efficiency, namely Pareto optimality and rank-maximality.

\begin{definition}[Pareto optimality]
	Given a preference profile $R = (R_1,\ldots,R_n)$, a matching $M$ is said to be Pareto optimal w.r.t.\ $R$ if no other matching can make an agent happier without making some other agent less happy, i.e., if 
	 $ \forall M' \in \mathcal{M}: \; \left(\exists a_i \in N, M'(a_i) \succ_{R_i} M(a_i) \right) \Rightarrow\\ \left(\exists a_j \in N, M'(a_j) \prec_{R_j} M(a_j)\right).$
\end{definition}

\begin{definition}[rank-maximality]\label{def:RMM}
Given a preference profile $R$ and a matching $M$, let the \textit{signature} of $M$ w.r.t. $R$, denoted $\sig_R(M)$, be the tuple $(x_1,\ldots,x_n)$, where $x_\ell$ is the number of $(a_i, o_j) \in M$ such that $o_j$ is the $\ell^{\text{th}}$ most preferred choice of $a_i$. We say that $(x_1,\ldots,x_n) \succ_{\text{sig}} (x'_1,\ldots,x'_n)$ if $\exists k \in [n]$ such that $x_k > x'_k$ and $x_i = x'_i$ for all $i < k$; this is the lexicographic comparison of signatures. A matching $M$ is rank-maximal w.r.t.\ $R$ if  $\sig_R(M)$ is maximum under $\succ_{\text{sig}}$. In words, a matching is rank-maximal if it matches max.\ number of agents to their top choice, subject to that the max.\ number of agents to their second choice, and so on.
\end{definition}

\subsection{Partial preferences}

In this paper, we study a partial preference model called \emph{top-$k$ preferences}, in which each agent $i$ may reveal a prefix of her preference ranking, that is, a ranking of her $k_i$ most favorite objects for some $k_i \in [n]$. Specifically, agent $i$ may reveal $P_i$ which is a linear order over some $L \subseteq O$ such that any object in $L$ are preferred to any object in $L \setminus O$ by the agent; we say that agent $i$ has revealed object $o$ in $P_i$ if $o \in L$. Also, we sometimes use $|P_i| = |L| = k_{i}$ to denote the size of $P_i$, i.e., the number of objects revealed in $P_i$.
We refer to $P = (P_1,\ldots,P_n)$ as a top-$k$ preference profile and say that a linear order $R_i$ is \textit{consistent} with $P_i$ if, for all $j, k \in [n]$, $o_j \succ_{P_i}  o_k \Rightarrow o_j \succ_{R_i}  o_k$. We use $\calC(P_i)$ to denote the set of linear orders over $O$ that are consistent with $P_i$, and $\calC(P)$ to denote $\calC(P_1) \times \ldots \times \calC(P_n)$.

Given a top-$k$ preference profile $P = (P_1, \cdots, P_n)$, we use  $\rev(P)$ to denote the set of all pairs $(a_i, o_j)$ such that $a_i$ has revealed object $o_j$ in $P_i$, and $\unrev(P)$ to denote $(N \times O) \setminus \rev(P)$. Throughout, whenever we refer to the size of a matching $M$ w.r.t.\ to $P$, we mean $|M \cap \rev(P)|$. 

We are interested in the notion of ``necessarily efficient'' matchings, which are guaranteed to satisfy an efficiency notion such as Pareto optimality or rank-maximality given only a top-$k$ preference profile~\citep{GNNW14,xia08,aziz2015possible}.

\begin{definition}[necessary Pareto optimality (NPO) and necessary rank-maximality (NRM)]
	Given a top-$k$ preference profile $P$, a matching $M$ is said to be necessarily Pareto optimal (resp.\ necessarily rank-maximal) w.r.t.\ $P$ if it is Pareto optimal (resp.\ rank-maximal) w.r.t.\ any profile of linear orders that is consistent with $P$, i.e., for any $R \in \calC(P)$.
\end{definition}

\subsection{Query model and elicitation} \label{sec:prelims-qm}
We consider algorithms that can elicit preferences by asking agents to reveal their ``next-best'' objects---meaning if the agent has already revealed their $k-1$ most preferred objects, then it asks the agent to reveal their $k$-th most preferred object. We formally define this query model below. 

\begin{definition}[next-best query model] \label{def:qmodel1}
	For $a_i \in N$ and $k \in [n]$, query $\mathcal{Q}(a_i, k)$, which asks agent $a_i$ to reveal the object in the $k$-th position in her preference list, is a valid query in the next-best query model if query $\mathcal{Q}(a_i,j)$ has already been made for all $j \in [k-1]$.
\end{definition}

We are interested in \textit{online} elicitation algorithms that find NPO or NRM matchings in the next-best query model and measure their performance against an all-powerful optimal algorithm. An all-powerful optimal algorithm is a hypothetical offline algorithm that can identify the minimum number of queries of the form $\mathcal{Q}(a_i, j)$ that are sufficient to guarantee an NPO  or NRM matching. That is, for $i \in [n]$ and $k_i \in [n]$, the optimal algorithm can minimize $\sum_{i \in [n]} k_i$, and can claim that asking every agent $i$ for their top-$k_i$ order, $P_i$, is sufficient to compute a matching that is NPO or NRM w.r.t.\ all consistent completions of $P = (P_1, \ldots, P_n)$. 
We say that an online elicitation algorithm is \textit{$\alpha$-competitive} (or, equivalently, it achieves a competitive ratio of $\alpha$) if the number of queries it requires to ask in the worst-case across all possible instances is at most $\alpha \cdot \text{OPT}$, where $\text{OPT}$ is the number of queries asked by the optimal algorithm. 

\section{Necessarily Pareto optimal matchings} \label{sec:npo}
We begin, in \Cref{sec:checkNPO}, by presenting an algorithm that determines if a given matching is NPO w.r.t.\ to a profile of top-$k$ preferences. Following this, we provide a useful characterization for when NPO matchings exist. This in turn is used in \Cref{sec:npo-elicit} in designing an elicitation algorithm that is $2(\sqrt{n} + 1)$-competitive. We also show that no elicitation algorithm can achieve a competitive ratio of $o(\sqrt{n})$, establishing our algorithm as asymptotically optimal.

\subsection{Determining whether a matching is NPO} \label{sec:checkNPO}
We present a polynomial time algorithm to determine whether a given matching $M$ is NPO w.r.t.\ a profile of top-$k$ preferences $P$. The algorithm proceeds in two steps: In Step 1, we build a directed graph $\bg=(N, E)$, where $E$ consists of directed edges $(a_i,a_j)$ such that $M(a_j) \succ_{R_i} M(a_i)$ under some $R_i \in \mc(P_i)$. That is, there exists an edge when agent $i$ prefers the object allocated to agent $j$ over her own under some completion of her partial preferences. In Step 2, we efficiently check if this directed graph admits a cycle. If it does, matching $M$ is not NPO; otherwise, it is NPO.

\begin{restatable}{theorem}{thmchecknpo} \label{thm:checkNPO}
Given a top-$k$ preference profile $P$ and a matching $M$, there is a polynomial time algorithm to check if $M$ is NPO w.r.t.\ $P$.
\end{restatable}

\begin{algorithm}[tb]
	\centering
	\noindent\fbox{%
		\begin{varwidth}{\dimexpr\linewidth-4\fboxsep-4\fboxrule\relax}
			\begin{algorithmic}[1]
				\footnotesize
				\Input Top-$k$ preference profile $P = (P_1, \ldots, P_n)$, set of agents $N$, set of objects $O$, and matching $M$.
				\Output YES if $M$ is NPO w.r.t. $P$, and NO otherwise. 
				\LeftComment{\textbf{Step 1.} Construct a directed graph $\bg =(N, E)$.}
				\State $E\gets \emptyset$
				\For{each $i \in [n]$}
					\If{$(a_i,M(a_i))\in\rev(P)$}
						\State For every agent $j\neq i$ such that $M(a_j) \succ_{P_i} M(a_i)$, add an edge $(a_i,a_j)$.
					\Else
						\State For every agent $j\neq i$, add an edge $(a_i,a_j)$.
					\EndIf
				\EndFor
				\LeftComment{\textbf{Step 2.} Check if $\bg$ has a cycle.}
				\If{$\bg = (N,E)$ has a cycle}
					\State \Return NO
				\Else
				    \State \Return YES
				\EndIf
			\end{algorithmic}
	\end{varwidth}}
	\caption{An algorithm to determine whether a given matching is NPO.}
	\label{algo:checkNPO}
\end{algorithm}

\begin{proof}
We claim that the algorithm described above---formally presented as Algorithm~\ref{algo:checkNPO}---is such an algorithm. First, let us prove its correctness. The key observation, which the reader can easily verify, is that the graph $\bg$ constructed in Step $1$ contains an edge $(a_i,a_j)$ if and only if there exists a completion $R_i \in \mc(P_i)$ under which $M(a_j) \succ_{R_i} M(a_i)$. This includes edges $(a_i,a_j)$ for which $M(a_j) \succ_{P_i} M(a_i)$ is already true under the top-$k$ preference order $P_i$. 

Suppose the algorithm returns NO. Then, $\bg$ has a cycle. Let $C$ be one such cycle. Construct a matching, $M'$, from $M$ as follows: if agent $i$ has an outgoing edge $(a_i,a_j) \in C$ in the cycle, set $M'(a_i)=M(a_j)$, and if agent $i$ is not part of cycle $C$, then set $M'(a_i)=M(a_i)$. It is easy to verify that $M'$ is a matching. 

Let $R = (R_1,\ldots,R_n) \in \mc(P)$ be a completion such that if agent $i$ has an outgoing edge $(a_i,a_j) \in C$ in the cycle, then $R_i \in \mc(P_i)$ satisfies $M'(a_i) = M(a_j) \succ_{R_i} M(a_i)$, and otherwise $R_i \in \mc(P_i)$ is chosen arbitrarily. Then, it is easy to see that for each agent $i$ who is a part of cycle $C$, we have $M'(a_i) \succ_{R_i} M(a_i)$, and for every remaining agent $i$, we have $M'(a_i) = M(a_i)$. Thus, $M'$ Pareto-dominates $M$ under the completion $R \in \mc(P)$, which implies that $M$ is not NPO w.r.t. $P$, as desired. 

Conversely, suppose $M$ is not NPO w.r.t. $P$. Hence, there exists a matching $M'$ and a completion $R \in \mc(P)$ such that $M'$ Pareto-dominates $M$ under $R$. We claim that $\bg$ must have a cycle, and the algorithm therefore must return NO. To see this, consider the graph $\bg' = (N,E')$ where, for $a_i \neq a_j$, $(a_i,a_j) \in E'$, if $M'$ assigns to agent $i$ the object that was assigned to agent $j$ under $M$, i.e., if $M'(a_i) = M(a_j) \neq M(a_i)$. Note that in $\bg'$, each agent $i$ either has exactly one incoming and one outgoing edge, or no incident edges. Hence, $\bg'$ is a union of disjoint cycles. In particular, $\bg'$ has at least one cycle. 

However, since $M'$ Pareto-dominates $M$ under $R$, for every edge $(a_i,a_j) \in E'$, we have that $M'(a_i) = M(a_j) \succ_{R_i} M(a_i)$. Hence, by the observation above, edge $(a_i,a_j)$ must also exist in $\bg$. Thus, $\bg'$ is a subgraph of $\bg$. Since $\bg'$ has at least one cycle, so does $\bg$, as desired. 

Finally, Algorithm~\ref{algo:checkNPO} runs in polynomial time because constructing $\bg$ and checking whether $\bg$ has a cycle can be done in $O(n^2)$ time (the latter using depth-first search).
\end{proof}

At a high level, correctness of the algorithm follows from the observation that if there is a cycle in the directed graph constructed above, then one can construct a complete preference profile $R \in \mc(P)$ such that each agent in the cycle strictly prefers the object allocated to the next agent in the cycle under $M$. Then, trading objects along this cycle would result in a matching $M'$ that, under this completion $R$, Pareto-dominates $M$, establishing that $M$ is not NPO.  

\subsection{Computing an NPO matching, when it exists} \label{sec:npo-iff}
Next, we find a necessary and sufficient condition for a given top-$k$ preference profile $P$ to admit an NPO matching. In the case that an NPO matching exists, we also provide an algorithm to compute one in polynomial time. 

\begin{theorem} \label{thm:npo-iff}
    A top-$k$ preference profile $P$ admits an NPO matching if and only if there is a matching of size at least $n-1$ w.r.t.\ $P$. Moreover, there is a polynomial time algorithm to determine whether a top-$k$ preference profile admits an NPO matching, and to compute an NPO matching if it exists.
\end{theorem}
\begin{proof}
\noindent{($\Rightarrow$)}
 Consider a top-$k$ preference profile $P$ that admits an NPO matching $M$. Suppose for the sake of contradiction that there is no matching w.r.t.\ $P$ that is of size larger than $n-2$. Then, $|M\cap\rev(P)|\le n-2$, which implies that there are at least two agents, w.l.o.g.\ agents $a_1$ and $a_2$, who are matched to objects they have not revealed under $P$. It is easy to see that one can now construct a completion $R\in\mc(P)$ such that $M(a_2) \succ_{R_{1}} M(a_1)$ and $M(a_1) \succ_{R_{2}} M(a_2)$. Under this completion, exchanging the objects assigned to $a_1$ and $a_2$ would Pareto-dominate $M$, contradicting the fact that $M$ is NPO w.r.t. $P$.
 
\noindent{($\Leftarrow$)}
 Let $P$ be a top-$k$ preference profile such that there exists a matching of size at least $n-1$ w.r.t.\ $P$. Construct a weighted bipartite graph $\bg = (N \cup O, \rev(P))$, where each edge $(a_i, o_j)\in \rev(P)$ has weight equal to the rank of $o_j$ in $P_i$. Let $M$ be a matching of min.\ weight among all matchings of max.\ cardinality in $\bg$. Note that $|M|\ge n-1$. 
 
 Let us first consider the case when $|M|=n$. If so, then it is easy to see that $M$ is an NPO matching w.r.t. $P$. This is because any matching that Pareto-dominates $M$ in any completion $R \in \mc(P)$ must have size $n$ w.r.t. $P$ as well and a strictly lower weight than $M$, which is a contradiction. 
 
 Now, consider the case when $|M|=n-1$. So exactly one agent and one object are not matched. Without loss of generality, suppose these are $a_n$ and $o_n$. Construct a matching $\hat{M}$ of our problem by adding $(a_n,o_n)$ to $M$. Thus, $|\hat{M} \cap \rev(P)| = |M| = n-1$. We wish to show that $\hat{M}$ is NPO. Suppose for the sake of contradiction that there exists a matching $M'$ which Pareto dominates $\hat{M}$ w.r.t.\ some completion $R\in\mc(P)$. Hence, $M'$ must assign each agent an object she prefers (under $R$) at least much as what she receives under $\hat{M}$, and at least one agent a strictly better (under $R$) object than what she receives under $\hat{M}$. Thus, $M'$, when viewed as a matching in $\bg$, must have  cardinality $n-1$ and weight at most the weight of $M$. Further, weight of $M'$ is equal to the weight of $M$ only if $M'$ assigns each agent $a_i \neq a_n$ the same object that $\hat{M}$ assigns. However, any two matchings that differ must differ in the assignment to at least two agents. Hence, we have that $M'$ is a matching in $\bg$ with cardinality $n-1$ and weight strictly less than the weight of $M$, which is a contradiction. 
 

Now, we provide a polynomial-time algorithm to compute an NPO matching, when it exists. From the argument above, it is sufficient to be able to compute a min-weight max-cardinality matching $M$ of the graph $\bg$. This can be done in polynomial time using the Hungarian algorithm \cite{munkres1957algorithms}. Then, if $|M| = n$, the algorithm returns $M$, and if $|M| = n-1$, the algorithm returns $\hat{M}$ by adding the unmatched agent-object pair.
%
\end{proof}

\subsection{Elicitation to compute an NPO matching} \label{sec:npo-elicit}

In this section we turn to the question of elicitation. As mentioned in \Cref{sec:prelims-qm}, we use the next-best query model, and informally our main goal is to understand how to elicit as little information as possible in order to compute an NPO matching. Towards this end, we first prove, in \Cref{thm:lb-elicit-po}, a lower bound that shows that any elicitation algorithm has a competitive ratio of $\Omega(\sqrt{n})$. 
%

\begin{restatable}{theorem}{thmLBnpo} \label{thm:lb-elicit-po}
In the next-best query model, if there exists an $\alpha$-competitive elicitation algorithm for finding a necessarily Pareto optimal matching, then $\alpha \in \Omega(\sqrt{n})$.
\end{restatable}

\if\npoSecProofsinApp0
Although the main idea in proving the lower bound is fairly straightforward, the formal proof is somewhat notationally heavy. Therefore, before presenting the proof, we sketch the high-level idea. 

The proof proceeds by constructing a family of instances $\mathcal{F}$, where each instance has the following characteristics. First, the $n$ agents are partitioned into $\sqrt{n}$ blocks each consisting of $\sqrt{n}$ agents each. Within each block $j$, $\sqrt{n}-2$ of the agents have similarly-styled preferences. In particular, their preferences are such that for every such agent $a_i$, $a_i$ prefers object $o_i$ the most. Additionally, for all the agents in a block $j$, except one \textit{special} agent, objects $o_k$, where $k$ is a multiple of $\sqrt{n}$, appears in at least the $(n-\sqrt{n})$-th position in their preference list. The {special} agent in a block $j$ is called so since this agent has the object $o_{j\sqrt{n}}$ in the $\sqrt{n}$-th position in their preference list. 

Given the above, the main idea in the proof is to argue that any online algorithm has to find at least $(\sqrt{n}-1)$ of these {special} agents and that this in turn takes at least $\Omega(n\sqrt{n})$ queries. The proof of this fact employs an \textit{adversary} argument and can be informally summarized as follows: an adversary can keep ``hiding'' a special agent until an online algorithm makes $\Omega(n)$ queries, and since the algorithm has to find at least $(\sqrt{n}-1)$  {special} agents, this takes $\Omega(n\sqrt{n})$ queries. Once we have this, then the bound follows by observing that an optimal algorithm just needs to make $\mathcal{O}(n)$ queries for any instance in this family.

\begin{proof}[Proof of Theorem~\ref{thm:lb-elicit-po}]
For the ease of presentation, assume that $n$ is a perfect square. Let us define useful sets of indices. For all $j \in [\sqrt{n} ]$ and $t \in [\sqrt{n}-1]$, define 
\begin{align*}
N_j &= \left\{{(j-1)\sqrt{n}+1}, \cdots, {j\sqrt{n}-1}\right\}, \\
N^{-t}_j &= N_j \setminus \{ {(j-1)\sqrt{n}+t} \}.
\end{align*}
	
Additionally, define
\begin{align*}
S &= \{\sqrt{n}, 2\sqrt{n},\ldots, n\}, \\ 
U& = \{1,\ldots, n\}.
\end{align*}

For any set of indices $P$, we use $O_{P}$ to denote the set of all objects $o_k$ such that $k \in P$. Finally, when we use a set of objects in describing the preference ranking of an agent over objects, it should implicitly be replaced by a ranking of all objects in the set in the increasing order of their indices.  

We are now ready to define our family of instances $\mathcal{F}$. Fix $t_1,\ldots,t_{\sqrt{n}} \in [\sqrt{n}-1]$, which will parametrize a particular instance in $\mathcal{F}$. Let us denote the corresponding instance $I_{(t_1,\ldots,t_{\sqrt{n}})}$, under which agent preferences are as in \Cref{fig:pref} for all $j \in [\sqrt{n}]$.
\begin{figure*}[tb]
\begin{tabular}{p{0.4\textwidth}p{0.5\textwidth}}
For agent $a_k$, where $k \in N^{-t_j}_j$ & : $o_k \succ O_{N^{-k}_j} \succ O_{U\setminus (N_j \cup S)} \succ O_S$,\\
	For agent $a_{y}$, where $y = (j-1)\sqrt{n} + t_j$ & : $o_{y} \succ O_{N^{-y}_j} \succ o_{j\cdot \sqrt{n}} \succ O_{U\setminus (N_j \cup S)} \succ O_{S\setminus\{{j\cdot\sqrt{n}}\}}$,\\
	
	For agent $a_{j\cdot \sqrt{n}}$ & : $O_{N_j} \succ O_{U\setminus (N_j \cup S)} \succ  O_S$.
\end{tabular}
\caption{Agent preferences used in the proof of Theorem~\ref{thm:lb-elicit-po}}
\label{fig:pref}
\end{figure*}
\Cref{fig1} illustrates an example of this construction when $n= 16$ and $t_j=2$ for each $j$.

	

	\begin{figure*}[t!] 
		{\scriptsize
 			\begin{center}
				\begin{tabular}{|l|l|l|l|l|l|l|l|l|l|l|l|l|l|l|l|l|l|l|}
					\cline{1-4} \cline{6-9} \cline{11-14} \cline{16-19}
					$a_1$    & $a_2$    & $a_3$    & $a_4$    &  & $a_5$    & $a_6$    & $a_7$    & $a_8$    &  & $a_9$    & $a_{10}$ & $a_{11}$ & $a_{12}$ &  & $a_{13}$ & $a_{14}$ & $a_{15}$ & $a_{16}$ \\ \cline{1-4} \cline{6-9} \cline{11-14} \cline{16-19} 
					$o_1$    & $o_2$    & $o_3$    & $o_1$    &  & $o_5$    & $o_6$    & $o_7$    & $o_5$    &  & $o_9$    & $o_{10}$ & $o_{11}$ & $o_9$    &  & $o_{13}$ & $o_{14}$ & $o_{15}$ & $o_{13}$ \\
					$o_2$    & $o_1$    & $o_1$    & $o_2$    &  & $o_6$    & $o_5$    & $o_5$    & $o_6$    &  & $o_{10}$ & $o_9$    & $o_{9}$  & $o_{10}$ &  & $o_{14}$ & $o_{13}$ & $o_{13}$ & $o_{14}$ \\
					$o_3$    & $o_3$    & $o_2$    & $o_3$    &  & $o_7$    & $o_7$    & $o_6$    & $o_7$    &  & $o_{11}$ & $o_{11}$ & $o_{10}$ & $o_{11}$ &  & $o_{15}$ & $o_{15}$ & $o_{14}$ & $o_{15}$ \\
					$o_5$    & \red{$o_4$}    & $o_5$    & $o_5$    &  & $o_1$    & \red{$o_8$}    & $o_1$    & $o_1$    &  & $o_1$    & \red{$o_{12}$} & $o_1$    & $o_1$    &  & $o_9$    & \red{$o_{16}$} & $o_9$    & $o_9$    \\
					$o_6$    & $o_5$    & $o_6$    & $o_6$    &  & $o_2$    & $o_1$    & $o_2$    & $o_2$    &  & $o_2$    & $o_1$    & $o_2$    & $o_2$    &  & $o_{10}$ & $o_9$    & $o_{10}$ & $o_{10}$ \\
					$o_7$    & $o_6$    & $o_7$    & $o_7$    &  & $o_3$    & $o_2$    & $o_3$    & $o_3$    &  & $o_3$    & $o_2$    & $o_3$    & $o_3$    &  & $o_{11}$ & $o_{10}$ & $o_{11}$ & $o_{11}$ \\
					$o_9$    & $o_7$    & $o_9$    & $o_9$    &  & $o_9$    & $o_3$    & $o_9$    & $o_9$    &  & $o_5$    & $o_3$    & $o_5$    & $o_5$    &  & $o_1$    & $o_{11}$ & $o_1$    & $o_1$    \\
					$o_{10}$ & $o_9$    & $o_{10}$ & $o_{10}$ &  & $o_{10}$ & $o_9$    & $o_{10}$ & $o_{10}$ &  & $o_6$    & $o_5$    & $o_6$    & $o_6$    &  & $o_2$    & $o_1$    & $o_2$    & $o_2$    \\
					$o_{11}$ & $o_{10}$ & $o_{11}$ & $o_{11}$ &  & $o_{11}$ & $o_{10}$ & $o_{11}$ & $o_{11}$ &  & $o_7$    & $o_6$    & $o_7$    & $o_7$    &  & $o_3$    & $o_2$    & $o_3$    & $o_3$    \\
					$o_{13}$ & $o_{11}$ & $o_{13}$ & $o_{13}$ &  & $o_{13}$ & $o_{11}$ & $o_{13}$ & $o_{13}$ &  & $o_{13}$ & $o_7$    & $o_{13}$ & $o_{13}$ &  & $o_5$    & $o_3$    & $o_5$    & $o_5$    \\
					$o_{14}$ & $o_{13}$ & $o_{14}$ & $o_{14}$ &  & $o_{14}$ & $o_{13}$ & $o_{14}$ & $o_{14}$ &  & $o_{14}$ & $o_{13}$ & $o_{14}$ & $o_{14}$ &  & $o_6$    & $o_5$    & $o_6$    & $o_6$    \\
					$o_{15}$ & $o_{14}$ & $o_{15}$ & $o_{15}$ &  & $o_{15}$ & $o_{14}$ & $o_{15}$ & $o_{15}$ &  & $o_{15}$ & $o_{14}$ & $o_{15}$ & $o_{15}$ &  & $o_7$    & $o_6$    & $o_7$    & $o_7$    \\
					$o_4$    & $o_{15}$ & $o_4$    & $o_4$    &  & $o_4$    & $o_{15}$ & $o_4$    & $o_4$    &  & $o_4$    & $o_{15}$ & $o_4$    & $o_4$    &  & $o_4$    & $o_7$    & $o_4$    & $o_4$    \\
					$o_8$    & $o_8$    & $o_8$    & $o_8$    &  & $o_8$    & $o_4$    & $o_8$    & $o_8$    &  & $o_8$    & $o_4$    & $o_8$    & $o_8$    &  & $o_8$    & $o_4$    & $o_8$    & $o_8$    \\
					$o_{12}$ & $o_{12}$ & $o_{12}$ & $o_{12}$ &  & $o_{12}$ & $o_{12}$ & $o_{12}$ & $o_{12}$ &  & $o_{12}$ & $o_8$    & $o_{12}$ & $o_{12}$ &  & $o_{12}$ & $o_8$    & $o_{12}$ & $o_{12}$ \\
					$o_{16}$ & $o_{16}$ & $o_{16}$ & $o_{16}$ &  & $o_{16}$ & $o_{16}$ & $o_{16}$ & $o_{16}$ &  & $o_{16}$ & $o_{16}$ & $o_{16}$ & $o_{16}$ &  & $o_{16}$ & $o_{12}$ & $o_{16}$ & $o_{16}$ \\ \cline{1-4} \cline{6-9} \cline{11-14} \cline{16-19} 
				\end{tabular}
				\caption{\small Preferences of the agents (presented in columns) constructed in the proof of \Cref{thm:lb-elicit-po} illustrated for the case where $n=16$ and $t_j=2$ for each $j$. The special agents are $\{a_2, a_6, a_{10}, a_{12}\}$.}
				\label{fig1}
 			\end{center}
		}
	\end{figure*} 
	
In order to simplify terminology, let us say that agent $a_i$ belongs to $N_j \cup \{j\sqrt{n}\}$ if $i \in N_j \cup \{j\sqrt{n}\}$. Also, we refer to the agents in $N_j \cup \{j\sqrt{n}\}$ as the agents in \textit{block} $j$. Thus, there are $\sqrt{n}$ blocks of agents in total. Finally, we refer to agent $(j-1)\cdot \sqrt{n} + k$ as the $k$-th agent in block $j$, and agent $(j-1)\cdot\sqrt{n}+t_j$ as the \textit{special} agent in block $j$. \Cref{fig1} highlights the special agents for a particular instance.

Recall that our family $\mathcal{F}$ consists of $(\sqrt{n}-1)^{\sqrt{n}}$ instances that are  parametrized by indices $t_1,\ldots,t_{\sqrt{n}}$. These indices specify which agent in each block is the special agent. Thus:
\begin{equation*}
\mathcal{F} = \left\{ I_{(t_1, \cdots, t_{\sqrt{n}})} \mid t_1,\ldots,t_{\sqrt{n}} \in [\sqrt{n}-1] \right\}.
\end{equation*}
	
We are now ready to contrast online algorithms with the \emph{optimal algorithm}. By the optimal algorithm, we refer to an omniscient algorithm that knows the complete underlying preference profile, and on any instance, can make the minimum possible number of queries in the next-best query model to reveal an NPO matching. Let OPT denote the maximum number of queries made by the optimal algorithm across all instances in $\mathcal{F}$. 

\begin{claim}
	$\text{OPT} \leq 3n - 2\sqrt{n}$. 
\end{claim}
\begin{proof}
	Let us consider an arbitrary instance $I_{(t_1, \cdots, t_{\sqrt{n}})}$ in $\mathcal{F}$. Note that the optimal algorithm has access to the complete preferences of the agents. To the agents in each block $j$, the optimal algorithm can make the following queries. 
	\begin{itemize}
	    \item To the special agent (i.e. the $t_j$-th agent), it makes $\sqrt{n}$ queries. 
	    \item To agent $a_{j\sqrt{n}}$, it makes $\sqrt{n}$ queries.
	    \item To every other agent in block $j$, it makes a single query. 
	\end{itemize} 
	
	Note that the number of queries made to agents in block $j$ is at most $2\cdot \sqrt{n} + (\sqrt{n}-2)\cdot 1 = 3\sqrt{n}-2$. Since there are $\sqrt{n}$ blocks, the total number of queries made is at most $\sqrt{n} \cdot (3\sqrt{n}-2) = 3n-2\sqrt{n}$, as desired.
	
	Finally, we need to show that the preference elicited through these queries reveal an NPO matching. Using \Cref{thm:npo-iff}, we need to show that there is a matching of size at least $n-1$ w.r.t. the elicited preferences. In fact, we argue that there is a matching of size $n$ w.r.t. the elicited preferences. 
	
	Consider the following matching $M$. For every $j \in [\sqrt{n}]$ and $y = (j-1)\sqrt{n}+t_j$, let
\begin{align*}
    &M(a_y) = o_{j\sqrt{n}}, \qquad M(a_{j\sqrt{n}}) = o_y,\\
	    &M(a_k) = o_k \text{ for every } k \in N_j^{-t_j}.
\end{align*}

	This matching has size $n$ w.r.t. the elicited preferences because $\sqrt{n}$ queries made to agents $a_y$ and $a_{j\sqrt{n}}$ reveal objects $o_{j\sqrt{n}}$ and $o_y$ in their preferences, respectively, and the single query made to every other agent $a_k$ in block $j$ reveals her top choice $o_k$. This completes the proof of the claim. 
\end{proof}
	
The claim above showed that an omniscient algorithm can make linearly many queries and reveal an NPO matching on \emph{any} instance from $\mathcal{F}$. In contrast, the next claim shows that an online algorithm must make $\Omega(n \sqrt{n})$ queries to find an NPO matching on \emph{some} instance in $\mathcal{F}$. Formally, fix an online algorithm, and let ALG denote the maximum number of queries it makes across all instances in $\mathcal{F}$.

\begin{claim} \label{clm:qm1-lb2}
	$\text{ALG} \geq  (\sqrt{n}-1)(n-2\sqrt{n})$.
\end{claim}

\begin{proof}
	For each $j \in [\sqrt{n}]$, let $t_j$ be defined as follows.
	\begin{itemize}
        \item If the algorithm makes at least $\sqrt{n}$ queries to all agents in block $j$, the $t_j$-th agent in block $j$ is the last agent in block $j$ to receive the $\sqrt{n}$-th query.  
        \item Otherwise, it is an arbitrary value in $[\sqrt{n}-1]$ such that the $t_j$-th agent in block $j$ has not been made $\sqrt{n}$ queries. 
    \end{itemize}

	Consider the corresponding instance $I_{(t_1,\ldots,t_{\sqrt{n}})}$. Note that this is an adversarially chosen instance depending on the way the online algorithm makes queries. This is well-defined because until the algorithm makes the $t_j$-th agent in block $j$ its $\sqrt{n}$-th query, it cannot distinguish which agent in block $j$ is the special agent; hence, if all agents in block $j$ receive at least $\sqrt{n}$ queries, the adversary can make the last agent in block $j$ to receive the $\sqrt{n}$-th query the special agent. 
	
	When the online algorithm terminates, it returns an NPO matching. Hence, by \Cref{thm:npo-iff}, there must be a matching, say $M$, of size at least $n-1$ w.r.t. the elicited preferences at termination. In particular, $M$ must match at least $\sqrt{n}-1$ of the $\sqrt{n}$ objects in $O_S$ to distinct agents who have revealed them. We use this to show that the algorithm must have made at least  $(\sqrt{n}-1)(n-2\sqrt{n})$ queries on this instance. 

	Let $O^* \subseteq O_S$ denote the set of objects in $O_S$ who are matched to agents who reveal them; note that $|O^*| \ge \sqrt{n}-1$ (i.e. $|O^*| \in \{\sqrt{n}-1,\sqrt{n}\}$). For each $o_{j\sqrt{n}} \in O^*$, we consider two cases.
	\begin{enumerate}
	    \item $o_{j\sqrt{n}}$ is matched to the special agent in block $j$. In this case, the algorithm must have made the $t_j$-th agent in block $j$ at least $\sqrt{n}$ queries. Thus, by the construction of the adversarial instance, it must have made each of $\sqrt{n}$ agents in block $j$ at least $\sqrt{n}$ queries. In this case, we ``charge'' these $\sqrt{n} \cdot \sqrt{n} = n$ queries made by the algorithm to object $o_{j\sqrt{n}}$.
	    \item $o_{j\sqrt{n}}$ is matched to an agent other than the special agent in block $j$. Note that this agent may not be in block $j$. Since any such agent puts $o_{j\sqrt{n}}$ in a position later than $n-\sqrt{n}$, the algorithm must have made at least $n-\sqrt{n}$ queries to this agent. We ``charge'' the last $n-2\sqrt{n}$ of these queries to object $o_{j\sqrt{n}}$. Note that we are not charging the first $\sqrt{n}$ queries made to this agent who may have already been charged to a different object through the first case above.
	\end{enumerate}
	
	Hence, each object in $O^*$ is charged at least $n-2\sqrt{n}$ unique queries. Since $|O^*| \ge \sqrt{n}-1$, the algorithm must have made at least $(\sqrt{n}-1) \cdot (n-2\sqrt{n})$ queries, as desired.
	\end{proof}
	
	

Finally, using the two claims above, we have that
\begin{align*}
\frac{\text{ALG}}{\text{OPT}} &\geq \frac{(\sqrt{n}-1)(n-2\sqrt{n})}{ 3n - 2\sqrt{n}}
\in \Omega(\sqrt{n}).
\end{align*}

Hence, the competitive ratio of any online algorithm is $\Omega({\sqrt{n}})$.	
\end{proof}
\fi

A natural strategy to elicit preferences is to simultaneously ask all the $n$ agents to reveal their next best object until the point at which the designer has enough information to compute an NPO matching. Recall that we can check the latter condition at every step by leveraging our Theorem~\ref{thm:npo-iff} from Section~\ref{sec:npo-iff}. Although this approach is natural, it is not difficult to see that it results in a competitive ratio of $\Omega(n)$. 

Let us describe the high-level idea of our algorithm (Algorithm~\ref{algo:elicit}), although the exact details vary slightly in the algorithm in order to optimize the competitive ratio. We follow the aforementioned na\"{\i}ve approach of asking all agents to report their next best object only up to a point where we can match at least $n-\sqrt{n}$ agents to objects they have revealed. At that point, we compute one such matching, focus on the unmatched agents (at most $\sqrt{n}$ of them), and then just ask these agents for their next best object until there is a matching of size at least $n-1$, which, by \Cref{thm:npo-iff}, is a sufficient condition for the existence of an NPO matching. Once we know an NPO matching exists, we use the polynomial time algorithm from \Cref{thm:npo-iff} to find one. This strategy, with the details optimized, results in an improved competitive ratio of $2(\sqrt{n} + 1)$, which, given Theorem~\ref{thm:lb-elicit-po}, is asymptotically optimal.

\begin{algorithm}[tb]
	\centering 
	\noindent\fbox{%
		\begin{varwidth}{\dimexpr\linewidth-4\fboxsep-4\fboxrule\relax}
			\begin{algorithmic}[1]
				\footnotesize
				\Input A profile of top-$k$ preferences $P = (P_1, \ldots, P_n)$, a set of agents $N$, and a set of objects $O$.
				\Output An NPO matching
								
				\State for each $i \in [n]$, initialize $P_i = \emptyset$  
				\State $\mathcal{M} \leftarrow \emptyset$; $k \leftarrow 1$; $s \leftarrow | \mathcal{M} |$
				\While{$s < (n - 1)$} \label{algstep:l3}
				\If{$s \leq (n - 1) - \min(\{k-1, \sqrt{n}\})$} \label{algstep:l1}
				\For{each $i \in [n]$} 
					\State $o \leftarrow \mathcal{Q}(a_i, |P_i|+1)$  
					\State update $P_i$ by adding object $o$ at position $|P_i| + 1$
				\EndFor \label{algstep:l2}
				\ElsIf{$s > (n - 1) - \min(\{k-1, \sqrt{n}\})$}  \label{algstep:l4}
					\State $U \leftarrow$ the set of unmatched agents in $\mathcal{M}$ \label{algstep:l5} 
					\For{each $i \in U$} 
					\State $o \leftarrow \mathcal{Q}(a_i, |P_i|+1)$  
					\State update $P_i$ by adding object $o$ at position $|P_i| + 1$
					\EndFor
				\EndIf \label{algstep:l6}
				\State $k \leftarrow k+1$
				\State $\mathcal{M} \leftarrow$ max.\ matching in $\mathbb{G}\left(N \cup O,\rev(P) \right)$
				\State $s \leftarrow | \mathcal{M} |$ 
				\EndWhile
				\State use algorithm described in Theorem~\ref{thm:npo-iff} to return a matching that is NPO w.r.t.\ ${P} = (P_1, \cdots, P_n)$ 
			\end{algorithmic}
	\end{varwidth}}
	\caption{A $2(\sqrt{n} + 1)$-competitive elicitation algorithm}
	\label{algo:elicit}
\end{algorithm}



\begin{theorem} \label{thm:ub-elicit-po}
	Algorithm~\ref{algo:elicit} is a $2(\sqrt{n} + 1)$-competitive elicitation algorithm in the next-best query model for computing a necessarily Pareto optimal matching.
\end{theorem}

\if\npoSecProofsinApp1
In order to prove this result, we first introduce the following notation. We maintain a graph $\bg = (N \cup O,\rev(P))$, where $P$ denotes the top-$k$ preference profile elicited so far. Let $s_j$ denote the size of the maximum matching in this graph during the $j$-th iteration of line~\ref{algstep:l3}. Next, let $m$ denote the number of times lines~\ref{algstep:l1}--\ref{algstep:l2} are run in Algorithm~\ref{algo:elicit}. Note that $m \geq 1$ (since $s_1 = 0$ and hence lines \ref{algstep:l1}--\ref{algstep:l2} are executed at least during the first iteration), and that the total number of iterations of the algorithm is $m+1$ (since the `else if', i.e., lines~\ref{algstep:l4}--\ref{algstep:l6} is executed just once). Also, note that for all $j \in [m]$, $s_j$ is also the size of the maximum matching when all agents have revealed their top $(j-1)$ objects in Algorithm~\ref{algo:elicit}. Finally, for $j \in [m]$, let $X_j$ denote the number of agents who are asked at least $j$ queries by the optimal algorithm. Now, we have the following claim, whose proof appears in \Cref{app:sec:proofs:clm1}. 

\else
 In order to prove this result, we first introduce the following notation and prove the claim below. Let $s_j$ denote the size of the maximum matching during the $j$-th iteration of line~\ref{algstep:l3}. Note that for all $j \in [m]$, $s_j$ is also the size of the maximum matching when all agents have revealed their top $(j-1)$ objects in Algorithm~\ref{algo:elicit}. Next, let $m$ denote the number of times lines~\ref{algstep:l1}--\ref{algstep:l2} are run in Algorithm~\ref{algo:elicit}. Note that $m \geq 1$, since $s_1 = 0$ and hence lines \ref{algstep:l1}--\ref{algstep:l2}, are executed at least during the first iteration, and that the total number of iterations of the algorithm is $m+1$ (since the `else if', i.e., lines~\ref{algstep:l4}--\ref{algstep:l6} is executed just once). Finally, for $j \in [m]$, let $X_j$ denote the number of agents who are asked at least $j$ queries by the optimal algorithm. Given this notation, we have the following claim about $\text{X}_j$.
\fi

\begin{claim} \label{clm:qm1-ub}
	For every $j \in [m]$, $X_j \geq (n-s_j-1)$.
\end{claim}

\if\npoSecProofsinApp0
\begin{proof}
	Suppose for the sake of contradiction that there exists $j \in [m]$ such that $X_j < n-s_j-1$. This implies that the remaining $n-X_j$ agents have been asked at most $j-1$ queries by the optimal algorithm. Let $L$ denote the set of these $n-X_j$ agents. 
	
	Now, by \Cref{thm:npo-iff}, the optimal algorithm must elicit preferences that admit a matching of size at least $n-1$. In such a matching, at least $n-X_j-1$ of the agents in $L$ must be matched to an object they have revealed, i.e., to one of their top $j-1$ objects. Thus, during the $j$-th iteration of line~\ref{algstep:l3}, there was a matching of size at least $n-X_j-1 > s_j$. This contradicts the fact that $s_j$ is defined as the size of the maximum matching during the $j$-th iteration. 
\end{proof}
\fi

Equipped with the notations and the claim above, we can now prove our theorem.

\begin{proof}[Proof of Theorem~\ref{thm:ub-elicit-po}]
	
	Let $\text{ALG}$ denote the number of queries asked in Algorithm~\ref{algo:elicit} and $m$ be as defined above. Since the algorithm asks $n$ queries during each time lines~\ref{algstep:l1}--\ref{algstep:l2} is run and it asks at most $(n-s_{m+1}) \cdot (n-m)$ queries during the execution of lines~\ref{algstep:l4}--\ref{algstep:l6} (the first term is maximum number of agents in set $U$ in line~\ref{algstep:l5} and the second arises from the fact that all these agents have revealed their top $m$ preferences), we have that
	\begin{align}
	\text{ALG} &\leq n\cdot m + (n-s_{m+1}) \cdot (n-m) \nonumber \\ 
	&< n\cdot m + (\min(\{m, \sqrt{n}\}) + 1) \cdot (n-m), \label{eq:qm1-1}
	\end{align}
	where the second step follows from the fact that $(n - s_{m+1}) < (1+\min(\{m, \sqrt{n}\}))$ (see line~\ref{algstep:l4} in Algorithm~\ref{algo:elicit}).
	
	Next, let $\text{OPT}$ denote the number of queries asked by the optimal algorithm. From Claim~\ref{clm:qm1-ub} we know that $X_j$, which is the number of agents who are asked at least $j$ queries by the optimal algorithm, is at least $(n-s_j-1)$. Therefore,
	\begin{align}
		\text{OPT} &\geq m\cdot X_m + \sum_{j=1}^{m-1} (m-j) \cdot (X_{m-j} - X_{m-j+1}) \nonumber \\
		&= X_m + X_{m-1} + \cdots + X_1 \nonumber \\
		&\geq n-1 + \sum_{j=2}^m (n-s_j-1) \tag{{\footnotesize using Claim~\ref{clm:qm1-ub} and the fact that $s_1 = 0$}} \nonumber  \\
		&\geq n-1 + (m-1)\cdot (n-s_m-1) \tag{{\footnotesize since $\forall j\in[m], s_j \leq s_{j+1}$}} \nonumber  \\ 
		&\geq n-1 + (m-1) \cdot \min(\{m-1, \sqrt{n}\}), 		\label{eq:qm1-2}
	\end{align}
		where the final inequality follows from the fact that $(n - 1 - s_{m}) \geq  \min(\{m-1, \sqrt{n}\})$ (see line~\ref{algstep:l1} in Algorithm~\ref{algo:elicit}).
Therefore, using~(\ref{eq:qm1-1}) and~(\ref{eq:qm1-2}), we have
\begin{align*}
\frac{\text{ALG}}{\text{OPT}} &\leq \frac{n\cdot m + (\min(\{m, \sqrt{n}\}) + 1) \cdot (n-m)}{n-1 + (m-1) \cdot \min(\{m-1, \sqrt{n}\})}\\
&\leq 2(\min(\{m, \sqrt{n}\}) + 1),
\end{align*}
where the last inequality follows by using $n \geq 2$, $m \geq 1$, and by considering the cases $m\leq \sqrt{n}$ and $m> \sqrt{n}$. 
\end{proof}

\section{Necessarily rank-maximal matchings}\label{sec:nrm} 
In the previous section, we devised algorithms for computing an NPO matching when one exists and when given a top-$k$ preference profile, and for eliciting a small amount of information to determine such matchings. 
Here, we focus on rank-maximality, which is a stronger notion of economic efficiency than Pareto optimality and widely used in real-world applications (see, e.g., \cite{garg2010assigning,abraham2006assignment,abraham2009matching}). To illustrate this notion, consider the top-$k$ preference profile $P$ with three agents, where $P_1=o_1\succ o_2 \succ o_3$, $P_2=o_1\succ o_2$ and $P_3=o_1$. The matching $M=\set{(a_1,o_3),(a_2,o_2),(a_3,o_1)}$ is NPO, but not NRM. To see this, consider the completion $R$ of $P$, where $R_1=o_1\succ o_2 \succ o_3$, $R_2=o_1\succ o_2 \succ o_3$, and $R_3=o_1\succ o_3 \succ o_2$. While $M$ has signature $(1,1,1)$ under $R$, matching $M'=\set{(a_1,o_1),(a_2,o_2),(a_3,o_3)}$ has a better signature $(1,2,0)$ under $R$. In fact, it is easy to verify that $M'$ is NRM w.r.t.\ $P$.

We begin, in \Cref{sec:determinNRM}, by presenting an algorithm that determines whether a given matching is necessarily rank-maximal (NRM) w.r.t.\ a given top-$k$ preference profile. In \Cref{sec:existNRM}, we build upon this algorithm to design another algorithm that decides if a given top-$k$ preference profile admits an NRM matching, and computes one if it does. Finally, we turn to the elicitation question, where we provide a lower bound on the competitive ratio of any elicitation algorithm that returns an NRM matching. We first introduce necessary additional definitions and notations.

\begin{definition}[rank-maximality w.r.t.\ top-$k$ preferences]
Given a top-$k$ preference profile $P$ and a matching $M$, let the \textit{signature} of $M$ w.r.t.\ $P$, denoted $\sig_P(M)$, be the tuple $(x_1,\ldots,x_n)$, where $x_\ell$ is the number of $(a_i, o_j) \in M\cap\rev(P)$ such that $o_j$ is the $\ell^{\text{th}}$ most preferred choice of $a_i$. We say that $(x_1,\ldots,x_n) \succ_{\text{sig}} (x'_1,\ldots,x'_n)$ if there exists a $h \in [n]$ such that $x_h > x'_h$ and $x_i = x'_i$ for all $i < h$; this is the lexicographic comparison of signatures. A matching $M$ is rank-maximal w.r.t.\ $P$ if  $\sig_P(M)$ is maximum under $\succ_{\text{sig}}$. In plain words, a matching is rank-maximal w.r.t.\ partial preference profile $P$ if it maximizes the number of agents matched to their top choice, subject to that the number of agents matched to their second choice, and so on.
\end{definition}
 
\begin{definition}[Extended signature]  Following Definition~\ref{def:RMM}, we define the \textit{extended signature} of a matching $M$ w.r.t.\ $P$, denoted $\extsig_P(M)$, to be the tuple $(x_1,\ldots,x_n)$, where, for all $i \in [n-1]$, $x_i$ is defined as in $\sig_P(M)$, while $x_n$ counts not only the number of agents matched to their $n^\text{th}$ choice revealed in $P$, but also the number of agents matched to an object they did not reveal in $P$.
\end{definition} 

\begin{definition}[Optimal signature] Given a top-$k$ preference profile $P$ over a subset of agents $S \subseteq N$ and a subset of objects $T \subseteq O$, as well as a set of ``forbidden'' pairs $F \subseteq S \times T$, let $\calM(S,T,F)$ denote the set of all matchings $M$ that match agents in $S$ to objects in $T$ and satisfy $M \cap F = \emptyset$. We define $\optsig_P(S,T,F) = \sup_{M \in \calM(S,T,F), R \in \calC(P)} \sig_R(M)$. Informally, this is the \textit{best possible signature} that any matching can achieve under any completion of $P$ while avoiding pairings from $F$. We use $\arg\optsig_P(S,T,F)$ to denote the matching $M \in \mm(S, T, F)$ that obtains this best signature under some completion $R \in \calC(P)$. Note that sometimes we use $\optsig_P$ to denote $\optsig_P(N,O,\emptyset)$. 
\end{definition}

\subsection{Determining whether a matching is NRM} \label{sec:determinNRM}

At a high level, our algorithm to determine if a given matching $M$ is NRM under a given top-$k$ preference profile $P$ works as follows. It leverages the fact that given any $S \subseteq N$, $T \subseteq O$, and $F \subseteq S \times T$, we can efficiently compute $\optsig_P(S, T, F)$. \Cref{algo:optsig} formally describes how to do this and the following lemma shows its proof of correctness. However, before that we make the following remark.

\textbf{Remark:} \citet{irving2006rank} define rank-maximal matchings more generally for weak preference orders, in which several objects can be tied at the same position in the preference ranking of an agent. The signature of a matching with respect to a weak order profile still encodes, for all $\ell$, the number of agents who are matched to one of the objects in the $\ell^{\text{th}}$ position in their preference list. \citet{irving2006rank} construct an edge-labeled bipartite graph $\bg$, where edge $(a_i,o_k)$ has label $\ell$ if $o_k$ appears in position $\ell$ in agent $a_i$'s weak preference order. By doing so, they extend the notion of rank-maximal matchings to such edge-labeled bipartite graphs. We note that their efficient algorithm to compute a rank-maximal matching for such a graph remains valid even when some of the edges of the graph are deleted, as we do in Algorithm~\ref{algo:optsig}.

\begin{restatable}{lemma}{lemmOptsig} \label{lemm:optsig}
	Given a top-$k$ preference profile $P$ over a set of agents $S\subseteq N$ and a set of objects $T \subseteq O$, and a set of pairs $F \subseteq S \times T$, Algorithm~\ref{algo:optsig} computes $\optsig_P(S,T,F)$ and $\arg \optsig_P(S,T,F)$ in polynomial time.
\end{restatable} 
\begin{proof}
	First, note that Algorithm~\ref{algo:optsig} runs in polynomial time because the algorithm of \citet{irving2006rank} runs in polynomial time. Hence, we only argue the correctness of our algorithm below. 
	
	Let $\{k_i : a_i \in S\}$, $P'$, and $\bg$ be as described in Algorithm~\ref{algo:optsig}. Recall that $\mm(S,T,F)$ is the set of matchings that match agents in $S$ to objects in $T$ without using any pairs of agent-object from $F$. 
	
	First, we claim that for any matching $M \in \mm(S,T,F)$, $\sig_{P'}(M) = \sup_{R \in \mc(P)} \sig_R(M)$. To see this, note that the contribution of all pairs $(a_i,o_k) \in M \cap \rev(P)$ to the signature $\sig_R(M)$ is fixed regardless of $R$. Thus, the best signature is achieved under the completion $R$ in which, for every $(a_i,o_k) \in M \cap \unrev(P)$, object $o_k$ appears at position $k_i+1$ in $a_i$'s preference ranking. This is precisely what occurs under the weak order profile $P'$ as well. Hence, $\sig_{P'}(M) = \sup_{R \in \mc(P)} \sig_R(M)$. 
	
	Taking supremum over $M \in \mm(S,T,F)$ on both sides, we have that 
	\[
	\sup_{M \in \mm(S,T,F)} \sig_{P'}(M) = \sup_{M \in \mm(S,T,F)} \sup_{R \in \mc(P)} \sig_R(M).
	\]
	By definition, our algorithm computes the LHS, while the RHS is equal to $\optsig_P(S,T,F)$.
\end{proof}

Given the result above, our algorithm first checks if $M$ has size at least $n-1$ w.r.t.\ $P$; this is necessary for $M$ to even be NPO w.r.t.\ $P$, as the proof of \Cref{thm:npo-iff} shows. If $M$ has size $n$ w.r.t.\ $P$, then we argue that its signature w.r.t.\ $P$ (which is the same as its signature w.r.t.\ any completion $R$ of $P$) must be $\optsig_P$. On the other hand, suppose $M$ has size $n-1$ w.r.t.\ $P$, with $(a_i,o_j) \in M \cap \unrev(P)$ being the pair where an agent is matched to an object she did not reveal under $P$. Then, we argue that two conditions need to be met: (a) to ensure that another matching $M'$ with $(a_i,o_j) \in M'$ cannot have a better signature under any completion $R$, we require that the signature of $M$ w.r.t.\ $P$ be at least as good as $\optsig_P(N\setminus\set{a_i},O\setminus\set{o_j},\emptyset)$, and (b) to ensure that another matching $M'$ with $(a_i,o_j) \notin M'$ cannot have a better signature under any completion $R$, we require that the extended signature of $M$ w.r.t.\ $P$ be at least as good as $\optsig_P(N,O,\set{(a_i,o_j)})$. \Cref{algo:checkNRM} formalizes this and we show  its correctness below.

\begin{algorithm}[!tb]
	\centering
	\noindent\fbox{%
		\begin{varwidth}{\dimexpr\linewidth-4\fboxsep-4\fboxrule\relax}
			\begin{algorithmic}[1]
				\footnotesize
				\Input A top-$k$ preference profile $P$ over a set of agents $S\subseteq N$ and a set of objects $T \subseteq O$, and $F \subseteq S \times T$.
				
				\Output $\optsig_P(S,T,F)$ and $\arg \optsig_P(S,T,F)$
				
				\State For all $a_i \in S$, let $k_i \gets |P_i|$
				
				\State $P' \gets$ weak order profile which matches $P$ on the preferences revealed under $P$, and for every $(a_i, o_j) \in \unrev(P)$, has object $o_j$ at position $k_i+1$ in $P'_i$ \label{algo:optsig-l2}
				
				\State Create a bipartite graph $\mathbb{G} = (S \cup T, E)$, where $E = S \times T \setminus F$, and each edge $(a_i, o_j) \in E$ is labelled with $\ell$ if $o_j$ appears in position $\ell$ in $P'_i$. 

				\State $M \gets$ compute a rank-maximal matching of $\mathbb{G}$ using the algorithm by \citet{irving2006rank}				
				
				\State \Return $\sig_{P'}(M)$ and $M$
			\end{algorithmic}
		\end{varwidth}}
	\caption{Algorithm to compute $\optsig_P(S,T,F)$ and $\arg \optsig_P(S,T,F)$.}
	\label{algo:optsig}
\end{algorithm}

\begin{algorithm}[!tb]
	\centering
	\noindent\fbox{%
		\begin{varwidth}{\dimexpr\linewidth-4\fboxsep-4\fboxrule\relax}
			\begin{algorithmic}[1]
				\footnotesize
				\Input A top-$k$ preference profile $P$, a set of agents $N$, a set of objects $O$, and a matching $M$.
				
				\Output Is $M$ NRM w.r.t.\ $P$?
				\State $c \gets |M \cap \rev(P)|$.
				\If{$c = n$ and $\sig_P(M) \equivsig \optsig_P$}
				\State \Return YES \label{algo:checkNRM-l5}
				\ElsIf{$c = n-1$} 
				\State Pick $(a_i,o_j) \in M \cap \unrev(P)$ \Comment{this is unique since $M$ has cardinality $n-1$ w.r.t.\ $P$.}
				\If {$\sig_P(M) \succeq_{\text{sig}} \optsig_P(N\setminus \{a_i\}, O\setminus \{o_j\}, \emptyset)$ and $\extsig_P(M) \succeq_{\text{sig}} \optsig_{P}(N, O, \set{(a_i, o_j)})$}
				\State \Return YES \label{algo:checkNRM-l11}
				\EndIf
				\EndIf				
				\State \Return NO				
			\end{algorithmic}
	\end{varwidth}}
	\caption{Algorithm to check if a given matching is NRM given a top-$k$ preference profile.}
	\label{algo:checkNRM}
\end{algorithm}

\begin{restatable}{theorem}{thmCheckNRM} \label{thm:checkNRM}
	Given a top-$k$ preference profile $P$ and a matching $M$, there exists a polynomial-time algorithm that determines if $M$ is NRM w.r.t.\ $P$. 
\end{restatable}
\begin{proof}
	To prove this, we will show that Algorithm~\ref{algo:checkNRM} returns YES if and only if $M$ is NRM w.r.t.\ $P$.
	
	\proofcase{$(\Rightarrow)$} Suppose Algorithm~\ref{algo:checkNRM} returns YES. Then, one of two conditions is met: (1) $M$ is a perfect matching w.r.t. $P$ and $\sig_P(M) \equivsig \optsig_P$; we refer to this as Case~(a) below; or (2) $M$ is of size $n-1$ w.r.t. $P$, $\sig_P(M) \succeq_{\text{sig}} \optsig_P(N\setminus \{a_i\}, O\setminus \{o_j\}, \emptyset)$, and $\extsig_P(M) \succeq_{\text{sig}} \optsig_{P}(N, O, \set{(a_i, o_j)})$; we refer to this as Case~(b) below. Next, we argue that in both cases, $M$ is NRM w.r.t. $P$.
	
	\textbf{Case (a).} Since $M$ is a perfect matching w.r.t. $P$, note that, for all $R \in \calC(P)$, $\sig_P(M) \equivsig \sig_R(M)$. Combined with the fact that $\optsig_P \equivsig \sig_P(M)$, we have that $\optsig_P = \sig_R(M)$ for all completions $R \in \mc(P)$. Since $\optsig_P$ is the best signature that any matching can have under any completion of $P$, this implies that $M$ is rank-maximal w.r.t. every completion of $P$, i.e., that it is NRM w.r.t. $P$. 	
	
	\textbf{Case (b).} Let $(a_i, o_j) \in M \cap \unrev(P)$.  Also, let $\mm_1 \subseteq \mm$ be the set of matchings that match $a_i$ to $o_j$ (i.e. $(a_i,o_j) \in M_1$ for every $M_1 \in \mm_1$), and $\mm_2 = \mm \setminus \mm_1$ be the set of matchings which match $a_i$ to an object other than $o_j$. 
		
	First, we show that for any $R \in\calC(P)$ and $M_1 \in \mm_1$, $\sig_R(M) \succeq_{\text{sig}} \sig_R(M_1)$. Suppose for contradiction that there exist $M_1 \in \mm_1$ and $R \in \mc(P)$ such that $\sig_R(M_1) \succ_{\text{sig}} \sig_R(M)$. Since $(a_i, o_j) \in M_1$ and $(a_i, o_j) \in M$, we have that $\sig_R(M'_1) \succ_{\text{sig}} \sig_R(M')$, where $M_1' = M_1 \setminus \{(a_i, o_j)\}$ and $M' = M \setminus \{(a_i, o_j)\}$. This in turn implies $\optsig_{P}(N \setminus \{a_i\}, O\setminus \{o_j\}, \emptyset) \succeq_{\text{sig}} \sig_R(M'_1) \succ_{\text{sig}} \sig_R(M') \equivsig \sig_{P}(M)$, where the first transition follows from the definition of $\optsig_{P}(N \setminus \{a_i\}, O\setminus \{o_j\}, \emptyset)$ and the last transition follows from the fact that $(a_i,o_j) \in \unrev(P)$ does not affect the signature of $M$ under $P$. However, this contradicts the fact that $\sig_P(M) \succeq_{\text{sig}} \optsig_{P}(N \setminus \{a_i\}, O\setminus \{o_j\}, \emptyset)$.
	
	Next, we show that for any $R \in\calC(P)$ and $M_2 \in \mm_2$, $\sig_R(M) \succeq_{\text{sig}} \sig_R(M_2)$.
	Suppose for contradiction that there exist $M_2 \in \mm_2$ and $R \in \mc(P)$ such that $\sig_R(M_2) \succ_{\text{sig}} \sig_R(M)$. This in turn implies 
	$\optsig_{P}(N, O, \set{(a_i, o_j)}) \succeq_{\text{sig}} \sig_{R}(M_2) \succ_{\text{sig}} \sig_R(M) \succeq_{\text{sig}} \extsig_P(M)$, where the first transition follows from the definition of $\optsig_{P}(N, O, \set{(a_i, o_j)})$ and the fact that $(a_i, o_j) \notin M_2$, and the last transition follows from the definition of $\extsig_P(\cdot)$. However, note that this contradicts the fact that $\extsig_P(M) \succeq_{\text{sig}} \optsig_{P}(N, O, \set{(a_i, o_j)})$.
	
	Finally, since we showed that $\sig_R(M) \succeq_{\text{sig}} \sig_R(M')$ for every $M' \in \mm_1 \cup \mm_2 = \mm$ and every $R \in\calC(P)$, we have that $M$ is rank-maximal w.r.t. every completion of $P$, i.e., that it is NRM w.r.t. $P$. 

	\proofcase{$(\Leftarrow)$} Suppose that $M$ is NRM w.r.t. $P$. Then, it must be NPO w.r.t. $P$ as well. Hence, as we show in the proof of \Cref{thm:npo-iff}, $M$ must have size at least $n-1$ w.r.t. $P$. So, we again have the following two cases. Either $M$ is a perfect matching w.r.t. $P$, which we refer to as Case~(a) below, or $M$ is of size $n-1$ w.r.t. $P$, which we refer to as Case~(b) below.
	
    \textbf{Case (a).} First, observe that $\sig_P(M) \equivsig \sig_R(M)$ for all $R \in \calC(P)$ because $M$ is a perfect matching w.r.t. $P$. Second, since $M$ is NRM w.r.t. $P$, it is rank-maximal under every completion of $P$. Hence, $\sig_P(M)$ is the best signature any matching can have under any completion of $P$, i.e., $\sig_P(M) \equivsig \optsig_P$. Thus, Algorithm~\ref{algo:checkNRM} must return YES in line~\ref{algo:checkNRM-l5}.

	\textbf{Case (b).} Let $(a_i, o_j) \in M \cap~\unrev(P)$. Next, let us consider $M_1 = \arg\optsig_{P}(N \setminus \{a_i\}, O\setminus \{o_j\}, \emptyset)$ and let $R \in \calC(P)$ be the completion at which $M_1$ attains the signature $\optsig_{P}(N \setminus \{a_i\}, O\setminus \{o_j\})$. We will show that $\sig_P(M) \succeq_{\text{sig}} \optsig_{P}(N \setminus \{a_i\}, O\setminus \{o_j\}, \emptyset)$. To see this, let us assume that this is not the case and that $\optsig_{P}(N \setminus \{a_i\}, O\setminus \{o_j\}, \emptyset) \equivsig \sig_R(M_1) \succ_{\text{sig}} \sig_P(M)$. Next, let $M_1' = M_1 \cup \{a_i, o_j\}$, and let us consider the completion $R' \in \calC(P)$ which is the same as $R$, except that $o_j$ is moved to the $n^{\text{th}}$ position in the preference list of $a_i$. Since $M \cap \unrev(P) = \set{(a_i,o_j)}$, $(a_i,o_j) \in M_1$, and $\sig_R(M_1) \succ_{\text{sig}} \sig_P(M)$, one can see that $\sig_{R'}(M'_1) \succ_{\text{sig}} \sig_{R'}(M)$. Hence, $M$ is not rank-maximal w.r.t. $R'$, which contradicts the fact that $M$ is NRM w.r.t. $P$. Therefore, we have that
    \begin{equation}
    \sig_P(M) \succeq_{\text{sig}} \optsig_{P}(N \setminus \{a_i\}, O\setminus \{o_j\}, \emptyset). \label{eqn1:lemm-checkNRM}
    \end{equation}
		
	Finally, let us consider $M_2 =\arg\optsig_{P}(N, O, \set{(a_i, o_j)})$, and let $R \in \calC(P)$ be the completion at which $M_2$ attains the signature $\optsig_{P}(N, O, \set{(a_i, o_j)})$. Since $(a_i, o_j) \notin M_2$, consider the completion $R' \in \calC(P)$ which is the same as $R$, except object $o_j$ is moved to the $n^{\text{th}}$ position in the preference list of $a_i$. Note that $\sig_{R'}(M_2) \succeq_{\text{sig}} \sig_{R}(M_2) \equivsig \optsig_{P}(N, O, \set{(a_i, o_j)})$. However, since $M$ is NRM w.r.t. $P$, we also have $\extsig_P(M) = \sig_{R'}(M) \succeq_{\text{sig}} \sig_{R'}(M_2)$. Hence, we have 		
	\begin{equation}
	\extsig_P(M) \succeq_{\text{sig}} \optsig_{P}(N, O, \set{(a_i, o_j)}). \label{eqn2:lemm-checkNRM}
	\end{equation}
	
	Observations~(\ref{eqn1:lemm-checkNRM}) and~(\ref{eqn2:lemm-checkNRM}) together imply that  Algorithm~\ref{algo:checkNRM} will return YES in line~\ref{algo:checkNRM-l11}.
	
	This establishes the correctness of Algorithm~\ref{algo:checkNRM}. To see that the algorithm runs in polynomial time, it is sufficient to observe that Algorithm~\ref{algo:optsig} runs in polynomial time.
\end{proof}

\subsection{Computing an NRM matching, when it exists} \label{sec:existNRM}

The results from \Cref{sec:determinNRM} are key to devising an algorithm that determines if a given top-$k$ preference profile admits an NRM matching and computes one if it does. Specifically, note again that a matching that is NRM w.r.t.\ a top-$k$ preference profile $P$ must have size at least $n-1$ w.r.t.\ $P$. If there is an NRM matching of size $n$ w.r.t.\ $P$, then we argue that it must be rank-maximal w.r.t.\ $P$ and have signature $\optsig_P$. We can check this efficienty using the results from \Cref{sec:determinNRM}. If there is an NRM matching of size $n-1$ w.r.t.\ $P$ with $(a_i,o_j) \in M \cap \unrev(P)$, then we argue that $M \setminus \set{(a_i,o_j)}$ must be a rank-maximal matching of $N\setminus\set{a_i}$ to $O\setminus\set{o_j}$ w.r.t.\ $P$. By iterating over all agent-object pairs $(a_i, o_j)$, we can again check this efficiently. \Cref{algo:existNRM} formalizes this and the next result proves its correctness.

\begin{algorithm}[!t]
	\centering
	\noindent\fbox{%
		\begin{varwidth}{\dimexpr\linewidth-4\fboxsep-4\fboxrule\relax}
			\begin{algorithmic}[1]
				\footnotesize
				\Input A top-$k$ preference profile $P$, agents $N$, and objects $O$.
				
				\Output Returns an NRM matching if $P$ admits one, and returns NO otherwise.
				
				\State $M' \gets$ rank-maximal matching w.r.t.\ $P$   
				
				\If{$|M' \cap rev(P)| = n$ and $\sig_P(M) \equivsig \optsig_P$} \label{algo:existNRM-l2}
				\State \Return $M'$ 
				
				\Else
				
				\For{$(a_i,o_j) \in \unrev(P)$}
				
				
				
				\State $X \gets$ rank-maximal matching of $N \setminus \{a_i\}$ to $O \setminus\{o_j\}$ w.r.t.\ $P$ \label{algo:existNRM-l7}
				
				\State $M' \gets X \cup \set{(a_i,o_j)}$
				
				\State \textbf{if} {$M'$ is NRM w.r.t.\ $P$}, \textbf{then} return $M'$ \label{algo:existNRM-l9} 
				
				
				
				\EndFor
				\EndIf
				
				\State \Return NO
			\end{algorithmic}
	\end{varwidth}}
	\caption{Algorithm to check if an NRM matching exists given a top-$k$ preference profile.}
	\label{algo:existNRM}
\end{algorithm}

\begin{restatable}{theorem}{thmExistNRM} \label{thm:existNRM}
Given a top-$k$ preference profile $P$, there exists a polynomial-time algorithm that returns a necessarily rank-maximal matching if one exists, and returns NO otherwise.
\end{restatable}
\begin{proof}
We want to show that $P$ admits an NRM matching if and only if Algorithm~\ref{algo:existNRM} returns a matching (i.e. does not return NO). Moreover, in case that the algorithm returns a matching, we want to show that it is in fact NRM w.r.t.\ $P$. 

\proofcase{$(\Rightarrow)$} Let $M$ be an NRM matching w.r.t.\ $P$. Then, it must be NPO w.r.t.\ $P$. Hence, as we show in the proof of \Cref{thm:npo-iff}, $M$ must have size at least $n-1$ w.r.t.\ $P$. So, we again have the following two cases. Either $M$ is a perfect matching w.r.t.\ $P$, which we refer to as Case~(a) below, or $M$ is of size $n-1$ w.r.t.\ $P$, which we refer to as Case~(b) below.

\textbf{Case (a).} Observe that $\sig_P(M) \equivsig\sig_R(M)$ for all $R \in \calC(P)$ since $M$ is a perfect matching w.r.t.\ $P$. Second, note that $\sig_P(M) \equivsig \optsig_P$, since otherwise it follows from the definition of $\optsig_P$ that there is a matching $M'$ and a completion $R \in \calC(P)$ such that $\sig_{R}(M') \succ_{\text{sig}} \sig_R(M) = \sig_P(M)$, which in turn contradicts the fact that $M$ is NRM w.r.t.\ $P$. Finally, it is easy to see that since $M$ has size $n$ w.r.t.\ $P$ and it is NRM w.r.t. $P$, it has to be rank-maximal w.r.t.\ to $P$. Combining these observations, we have that any matching $M'$ that is a rank-maximal matching w.r.t. $P$ must have size $n$ and the same signature (under $P$) as $\optsig_P$. 
Hence, the check in line~\ref{algo:existNRM-l2} will succeed, and the algorithm will return matching $M'$, which is NRM w.r.t. $P$.  
	
\textbf{Case (b).} Let $(a_i,o_j) \in M \cap \unrev(P)$. As in Case~(a), we can argue that $M \setminus \set{(a_i,o_j)}$ is NRM w.r.t. $P$ among all matchings that match $N\setminus\set{a_i}$ to $O\setminus\set{o_j}$. Let $X$ be any matching of $N\setminus\set{a_i}$ to $O\setminus\set{o_j}$ that is rank-maximal w.r.t. $P$. Then, this implies that $X \cup \set{(a_i,o_j)}$ is also an NRM matching. Observe that this is precisely the condition that is checked in line~\ref{algo:existNRM-l9} of Algorithm~\ref{algo:existNRM}. Hence, the algorithm returns an NRM matching w.r.t. $P$ in this case too. 

\proofcase{$(\Leftarrow)$} Suppose $P$ does not have an NRM matching. Note that Algorithm~\ref{algo:existNRM} returns a matching only under two cases. First is when there is a matching $M'$ such that $M'$ is of size $n$ w.r.t.\ $P$ and has the same signature as $\optsig_P$. Note that if such a matching exists, then it is easy to see from the definition of $\optsig_P$ and the fact that $M'$ is a perfect matching w.r.t.\ $P$ that it has to be NRM w.r.t. $P$, which in turn contradicts our assumption that $P$ does not have an NRM. 

The other case is when Algorithm~\ref{algo:existNRM} returns a matching in line~\ref{algo:existNRM-l9}, and here, the proof trivially follows from \Cref{thm:checkNRM}, which proves the correctness of Algorithm~\ref{algo:checkNRM}.

Finally, one can see that \Cref{algo:existNRM} runs in polynomial time since Algorithms~\ref{algo:optsig} and~\ref{algo:checkNRM} run in polynomial time. 
\end{proof}


\subsection{Elicitation to compute an NRM matching}

We are now ready to consider the question of eliciting information to compute an NRM matching. Below we provide a lower bound on the competitive ratio of any elicitation algorithm.

\begin{restatable}{theorem}{thmLBnrm} \label{thm:nrm}
In the next-best query model, for $n \ge 2$, if there exists an $\alpha$-competitive elicitation algorithm for finding a necessarily rank-maximal matching, then $\alpha \ge \frac{4}{3} - \frac{2}{n}$, if $n$ is even, and $\alpha \ge \frac{4}{3} - \frac{20}{9n-3}$, if $n$ is odd. 
\end{restatable}

\begin{proof}
	Let us first consider the case when $n$ is even. Let us partition the set of agents $N$ and the set of objects $O$ in the following way. For $t \in [n/2]$, we define $N_t = \set{a_t,a_{n/2 + t}}$ and $O_t = \set{o_t,o_{n/2 + t}}$. Also, for $t \in [n/2]$, we refer to $N_t$ as a \textit{block}. Note that $N = \cup_{t=1}^{n/2} N_t$ and $O = \cup_{t=1}^{n/2} O_t$.
	
	Next, consider a family of instances $\mathcal{F}$ defined as follows. For each $t \in [n/2]$, both agents in $N_t$ have $o_t$ as their top choice; the second choice for one of the agents in $N_t$ (an agent who we will refer to as \textit{special} ) is $o_{n/2+t}$, while that of the other agent is $o_{y}$, where $y = (t+1)\mod(n/2)$. The rest of the preferences are filled in arbitrarily. Note that this family has $2^{n/2}$ instances. 
	
	Now, let us consider an arbitrary instance $I$ in $\mathcal{F}$. It is easy to verify that $I$ has a unique rank-maximal matching described as follows: for each $t \in [n/2]$, the agent in $N_t$ with $o_{n/2 + t}$ as her second choice is matched to $o_{n/2 + t}$, while the other agent in $N_t$ is matched to $o_t$ (her top choice). 
	
	Next, we argue that on this family of instances, any online algorithm for finding an NRM matching has to make at least $2n-3$ queries. To see this, recall that an NRM matching is also NPO, and therefore, must have, because of Theorem~\ref{thm:npo-iff}, a size at least $(n-1)$ w.r.t.\ the elicited preferences. This implies that the algorithm must make at least $2n-3$ queries, since in all the blocks $N_t$ except one, the algorithm must find an agent who reveals $o_{n/2 + t}$, but the first agent to whom the algorithm asks for her second choice may end up revealing $o_{y}$, where $y = (t+1) \mod (n/2)$, thus necessitating another query. In other words, in at least $(n/2-1)$ blocks, the algorithm must make at least $4$ queries each, and in the last block, the algorithm must make at least $1$ query, which adds up to $2n-3$. 
	
	In contrast, the offline optimal may ask all the $n/2$ special agents two queries each and all the other agents a single query, thus making a total of $3n/2$ queries; it is easy to verify that the rank-maximal matching in the underlying preferences will be a necessarily rank-maximal matching given such elicited preferences. 
	
	Therefore, combining all observations above, we have that, when $n$ is even, the competitive ratio of any online algorithm is at least $\frac{2n-3}{3n/2} = \frac{4}{3} - \frac{2}{n}$ .
	
	Finally, let us consider the case when $n$ is odd and let $n = n'+1$, where $n'$ is even. We can consider an agent $a^*$ and an object $o^*$ separately and let $o^*$ be the top choice of $a^*$ but the last choice of every other agent. We use the construction described above (for the case when we have even number of agents and objects) to construct the preferences for the other agents (i.e., every agent except $a^*$). Given this, it is easy to see that both the offline optimal and the online algorithm have one additional query to make (asking $a^*$ her top choice); it then follows that the offline optimal makes $\frac{3n'}{2}+1 = \frac{3n-1}{2}$ queries, while any online algorithm must make at least $2n'-3+1 = 2n-4$ queries. Hence, when $n$ is odd, the competitive ratio is at least $\frac{2n-4}{(3n-1)/2} = \frac{4}{3} - \frac{20}{9n-3}$.
\end{proof}

Note that our lower bound converges to $\frac{4}{3}$ as $n$ goes to infinity. So, a natural question is whether one can design an online algorithm to compute an NRM matching that has constant competitive ratio. We conjecture that this should be possible, but leave it as an open problem for future work. 
\section{Discussion}\label{sec:disc}

As discussed in \Cref{sec:nrm}, the most immediate open question that stems from our work is to design an online elicitation algorithm to compute an NRM matching, and analyze its competitive ratio. In particular, we believe that there exists an algorithm with a constant competitive ratio. 

Note that our online algorithm to compute an NPO matching from \Cref{sec:npo} is a constructive procedure, and uses our result that the existence of an NPO matching can be reduced to a simple analytical condition. Given that such a simple condition seems out of reach for the case of NRM, it may be interesting to explore a different approach to designing an online algorithm to compute an NRM matching. For example, one can train a machine learning model to select which agent to query next as a function of the preferences elicited so far. At each step, we can first use this algorithm to make one more query, and then call our efficient algorithm to check if the elicited preference reveal an NRM matching. It would be interesting to study if this approach can lead to a low competitive ratio, at least in practice.

More broadly, the literature on learning a desirable matching of agents to objects from partial preferences or historical data is still in its infancy, and many interesting directions, such as studying other models of partial preferences or other models of querying agents, remain unexplored.

\section*{Acknowledgments}
Hadi Hosseini acknowledges support from NSF grant \#1850076. Nisarg Shah was partially supported by an NSERC Discovery grant. 
We are grateful to Lirong Xia for his insights in shaping the initial core of this project.
We thank the anonymous reviewers for their very helpful comments and suggestions.

\printbibliography

\end{document}